\documentclass{article}
\usepackage{amsmath, amsthm}
\usepackage{url}
\usepackage{graphicx,psfrag,epsf}
\usepackage{enumerate}
\usepackage{natbib}
\usepackage{verbatim}

\setcitestyle{numbers,open={[},close={]}}

\usepackage[symbol]{footmisc}
\usepackage{epstopdf,gensymb}
\usepackage{graphics}
\usepackage{graphicx,color}
\newcommand{\erdos}{Erd\H{o}s-R\'enyi }
\usepackage{amsfonts, animate,subfigure}
\usepackage{amssymb}
\usepackage[margin=1in]{geometry}
\RequirePackage[colorlinks,citecolor=blue,urlcolor=blue]{hyperref}
\usepackage{todonotes}
\usepackage{setspace}
\usepackage{indentfirst}
\usepackage{float}
\usepackage{amsfonts}
\usepackage{bm,pbox}
\usepackage{multirow} 
\usepackage{color}
\usepackage[margin=1in]{geometry}
\usepackage{framed}
\usepackage{amssymb,amsthm,amsfonts,amsbsy,latexsym,amsmath}
\usepackage{enumerate}

\newenvironment{enumeratea}{\begin{enumerate}[\upshape (a)]}{\end{enumerate}}

\newtheorem{theorem}{Theorem}
\newtheorem{lemma}[theorem]{Lemma}

\newtheorem{definition}[theorem]{Definition}

\renewcommand{\geq}{\geqslant}

\begin{document}
	\begin{centering}
{\LARGE \bf Analysis of Population Functional Connectivity Data via Multilayer Network Embeddings}\\
\vskip 1pc
\author[{James D. Wilson$^*$\footnotetext{$^*$Corresponding Author. Email: jdwilson4@usfca.edu. Phone: (415) 422-6505}\\
		Department of Mathematics and Statistics\\
		University of San Francisco, San Francisco, CA 94117\\
		\vskip 1pc
        ~Melanie Baybay\\
 		Department of Computer Science\\
		University of San Francisco, San Francisco, CA 94117\\
		\vskip 1pc
        ~Rishi Sankar\\
 		Department of Computer Science\\
		University of California, Los Angeles, CA 90095
		\vskip 1pc
        ~Paul Stillman\\
		Department of Marketing\\
		Yale School of Management, New Haven, CT 06511\\
		\vskip 1pc
        ~Abbie M. Popa\\
 		The Data Institute\\
		University of San Francisco, San Francisco, CA 94117\\
		\vskip 2pc
		Action Editor: Filippo Menczer, Accepted for publication in Network Science
		}
\end{centering}






\begin{abstract}
Population analyses of functional connectivity have provided a rich understanding of how brain function differs across time, individual, and cognitive task. An important but challenging task in such population analyses is the identification of reliable features that describe the function of the brain, while accounting for individual heterogeneity. Our work is motivated by two particularly important challenges in this area: first, how can one analyze functional connectivity data over {populations} of individuals, and second, how can one use these analyses to infer group similarities and differences. Motivated by these challenges, we model population connectivity data as a multilayer network and develop the multi-node2vec algorithm, an efficient and scalable embedding method that automatically learns continuous node feature representations from multilayer networks. We use multi-node2vec to analyze resting state fMRI scans over a group of 74 healthy individuals and 60 patients with schizophrenia. We demonstrate how multilayer network embeddings can be used to visualize, cluster, and classify functional regions of the brain for these individuals. We furthermore compare the multilayer network embeddings of the two groups. We identify significant differences between the groups in the default mode network and salience network -- findings that are supported by the triple network model theory of cognitive organization. Our findings reveal that multi-node2vec is a powerful and reliable method for analyzing multilayer networks. Data and publicly available code is available at {https://github.com/jdwilson4/multi-node2vec}.
\end{abstract}

\tableofcontents

\section{Introduction}\label{sec:intro}
Human cognition is an emergent phenomenon of complex interactions among many different brain regions \cite{bressler2010large,medaglia2015cognitive, sporns2011networks,sporns2014contributions}. Network neuroscience is a common perspective of the brain in which neural connectivity is characterized through network-based models. Such network investigations have revealed general organizing principles of the whole brain, including high modularity \cite{sporns2016modular}, a ``rich-club" of interconnected hub regions \cite{van2011rich, van2012high}, and topologies that demonstrate small-world structure \cite{bassett2006small, achard2006resilient, bassett2006adaptive, he2007small}. These findings have shown, for instance, that the regions of the brain not only exhibit strong clustering, but also enable the brain to minimize wiring costs while maintaining robust transfer and integration of information across regions \cite{bullmore2012economy,fornito2011genetic}. Network investigations have also advanced our understanding of neural processes, such as learning and memory \cite{bassett2011dynamic,bassett2015learning}, cognitive control \cite{cole2012global}, and emotion \cite{kinnison2012network}. Investigations of local subnetwork structure of the brain have revealed consistent architectures that may describe overall functional efficiency \cite{stillman2017statistical, stillman2019consistent}. Several large-scale projects have arisen from network neuroscience, such as the Human Connectome Project \cite{van2012human,van2013wu}, as well as the BRAIN initiative \cite{insel2013nih}. 

Despite the many successes of network neuroscience in understanding the structure and function of the brain, many challenges remain. Our work is motivated by two particularly important challenges: (1) how can one analyze functional connectivity data over \emph{populations} of individuals, and (2) how can one utilize these analyses to infer group similarities and differences. To answer these two questions, we propose analyzing multilayer networks that effectively model the functional organization of each group. Population data of functional connectivity give rise to brain networks that are inherently multilayer -- they vary across time, across person, and across cognitive task \cite{betzel2016multi}. Unfortunately, many network neuroscience strategies are static and consider only a single-layered network representation of the brain. Single-layered analyses neglect heterogeneity among individuals as well as their interdependencies (see \cite{wilson2016community} for a discussion). Multilayer network representations of the brain enable researchers to fully analyze the relationships within and between networks observed over time, person, or task \cite{bassett2011dynamic, bassett2015learning}. 


Multilayer networks model the functional connections between regions of the brain across a population of individuals. Multilayer networks themselves are challenging data objects to analyze, and there is a lot of current research devoted to handling these challenges (see \cite{kivela2014multilayer} for a recent survey). In this paper, we propose a fast and scalable algorithm, called \emph{multi-node2vec}, that learns the nodal features from complex multilayer networks through the Skip-gram neural network model. By embedding multilayer networks of the brain to nodal features, we enable the direct analysis of the regions of the brain that are representative of the group under study.

We apply multi-node2vec to a multilayer brain network representing the functional connectivity of 74 healthy individuals and 70 patients with schizophrenia who underwent resting state fMRI. We demonstrate how to utilize the results of multi-node2vec for three primary objectives: (i) visualization and clustering of these regions into communities of similar features, (ii) classification of regions into anatomical regions of interest in the brain, and (iii) comparing two populations of individuals. We find that multi-node2vec identifies feature embeddings that closely match the functional organization of healthy individuals, and also provides a powerful strategy for comparing groups of individuals. Our proposed embedding technique provides a valuable step in automatically learning neurological variation among brains, including individual differences and disease. 


\subsection{Related Work}

Feature engineering is a common and important learning task in statistics and machine learning. Traditionally, feature engineering for networks, often referred to as network embedding, has amounted to manually describing summaries of networks based on a collection of user-selected network properties, like structural importance or subgraph counts \cite{gallagher2010leveraging, henderson2011s}. A similar strategy has been applied to multilayer networks, where chosen features attempt to quantify the within and between layer relationships among nodes \cite{boccaletti2014structure, kivela2014multilayer}. In contrast to these approaches, the multi-node2vec algorithm automatically learns important continuous features of multilayer networks and requires no user input on what properties to capture. 

Feature engineering techniques have been extensively used to identify low-rank representations of multivariate data. In this setting, the data matrix $\mathbf{X}$ is an $n \times p$ matrix whose rows are $n$ independent observations measured on $p$ features. Dimension reduction techniques are particularly important when the data is high-dimensional -- when $p > n$ -- as traditional statistical inference is often no longer reliable \cite{buhlmann2011statistics}. Singular value decompositions, principal components analysis, and spectral clustering, for instance, are well-studied decomposition techniques that have been applied to a number of high-dimensional problems ranging from topic modeling to micro-array analysis. These techniques rely on the spectral decomposition of $\mathbf{X}$, its empirical covariance matrix, and the graph Laplacian of a similarity matrix on the columns of $\mathbf{X}$, respectively. Though these methods are known to provide accurate representations of $\mathbf{X}$, they face a drawback in computational complexity for large $p$ due to matrix inversion, which can be prohibitive for especially high-dimensional problems. In Section 6, we show that multi-node2vec is in fact an approximation to closed-form implicit matrix factorization. 

There have been many feature learning techniques for static networks developed in the past decade. The latent space model from \cite{hoff2002latent}, for example, is a common model-based embedding technique that embeds the observed network onto Euclidean space - typically onto two-dimensions. Our current work is most closely related to the automatic feature learning techniques LINE \cite{tang2015line}, DeepWalk \cite{perozzi2014deepwalk}, and node2vec \cite{grover2016node2vec}. We briefly discuss these here but refer the reader to \cite{goyal2017graph} for a recent review of node embedding techniques for static networks. LINE, DeepWalk, and node2vec each learn features of a node from the neighborhoods of the node in the observed graph. LINE learns d-dimensional features by an objective function that preserves first and second order network properties. DeepWalk and node2vec each learn D-dimensional features using the Skip-gram neural network model, which minimizes a log-likelihood loss function that characterizes relationships from node neighborhoods in the observed graph. The Skip-gram  model was originally developed for learning efficient representations of words in a large document of text \cite{mikolov2013efficient, pennington2014glove}. The first application of the Skip-gram model was in the {word2vec} algorithm \cite{mikolov2013distributed}, where it was used to estimate a word's features through the log-likelihood cost minimization from the prediction of that word's context. DeepWalk, node2vec, and multi-node2vec differ in the way they collect node neighborhoods. DeepWalk extracts neighborhoods using truncated random walks. Node2vec performs second random walks based on hyperparameters that guide the likelihood of visiting nodes either closer to or further away from previously visited nodes. The development of our algorithm is motivated by the recent success of the node2vec algorithm on consensus matrices of structural MRI, \cite{rosenthal2018mapping}. The multi-node2vec algorithm, however, directly handles the analysis of populations of functional connectivity data. Multi-node2vec is also random walk based and can be thought of as a generalization of the original DeepWalk and node2vec algorithms. Utilizing Laplacian dynamics like that discussed in \cite{mucha2010community}, we incorporate a walk parameter that dictates the probability of moving from one layer to the next. 

Other recent work has focused on generative network models that model populations of networks, including the the random effects stochastic block model \cite{paul2018random}, the multi-subject stochastic block model \cite{pavlovic2019multi}, the hierarchical latent space model \cite{wilson2020hierarchical}, as well as the edge-based logistic model from \cite{simpson2019mixed}. These three models each assume independence of the edges within and across individuals. Even so, estimation methods for these models are sometimes prohibitive and typically require small network representations (on the order of 10s of nodes) for each individual.


The community detection task of partitioning the nodes of a multilayer network into densely connected subgroups, or communities, can be viewed as multilayer embedding. Specifically, the results of a community detection algorithm is an $N \times D$ matrix $\mathbf{F}$, where the $v$th row $\mathbf{f}_v$ is a binary vector that indicates which community(ies) the node $v$ is contained. The development of multilayer community detection methods is still in its early stages, but several useful techniques have been developed over the past decade \cite{de2015identifying, mucha2010community, stanley2015clustering, wilson2016community}. Though not the focus of this paper, it would be interesting to fully explore the use of communities as features for regression and other machine learning tasks in future work.

\section{Multilayer Embedding with multi-node2vec}

In resting state functional connectivity, network models are constructed by gauging the degree to which two regions' time-series activity is related to one another. The intuition is that the greater two regions are functionally connected, the more their time-series' should co-activate. In the present study, we model the strength of connection between two regions based on the correlation between the two regions' activity during a resting state fMRI scan (i.e., when participants have no task except to stay awake \cite{bullmore2009complex, smith2011network}. We can subsequently apply the multi-node2vec technique to identify local features of the brain from a group of individuals.

%

A multilayer network of length $m$ is a collection of networks or graphs $\{G_1, \ldots, G_m\}$, where the graph $G_{\ell}$ models the relational structure of the $\ell$th layer of the network. Each layer $G_\ell = (V_{\ell}, W_\ell)$ is described by the vertex set $V_\ell$ that describes the units, or actors, of the layer, and the \emph{intra-layer} edge weights $W_\ell = \{w_{\ell}(u,v): u, v \in V_\ell\}$ that describes the strength of relationship between the nodes. Furthermore, there is a collection of \emph{inter-layer} edge weights $IL := \{w_{\ell, \ell'}(u,v): u \in V_\ell, v \in V_{\ell'}\}$ that describe relationships between nodes of differing layers. Note that layers in the multilayer sequence may be heterogeneous across vertices, edges, and size. In the case of population studies of functional connectivity, each layer $G_\ell$ represents the correlation network arising from resting state fMRI for individual $\ell$. Denote the set of unique nodes in $\{G_1, \ldots, G_m\}$ by $\mathcal{N}$, and let $N = |\mathcal{N}|$ denote the number of nodes in that set. Throughout the remainder of this paper, to signify the unique node set $\mathcal{N}$ we represent multilayer networks with $m$ layers and node set $\mathcal{N}$ as $\mathbf{G}_{\mathcal{N}}^{m}$.

Multilayer networks are inherently complex and high-dimensional. Without further assumptions on $\mathbf{G}_{\mathcal{N}}^{m}$, inference on $\mathcal{N}$ necessitates the modeling of $N^2$ (possibly dependent) edge variables, which is computationally challenging even for moderately sized $N$. In light of this challenge, the aim of the current work is to learn an interpretable low-dimensional feature representation of the nodes in a multilayer network. In particular, we seek a $D$-dimensional representation

\begin{equation} \label{eq:lowdimf} 
\mathbf{F}: \mathcal{N} \rightarrow \mathbb{R}^D,
\end{equation}

\noindent where $D < < N$. The function $\mathbf{F}$ can be viewed as an $N \times D$ matrix whose rows $\{\mathbf{f}_v: v = 1, \ldots, N \}$ represent the feature space of each node in $\mathcal{N}$. 

\subsection{Maximum Likelihood Formulation}
Let $\mathbf{G}_{\mathcal{N}}^{m}$ be an observed multilayer network with $m$ layers and the set of unique nodes $\mathcal{N}$. Our aim is to learn $D$ representative features of $\mathcal{N}$ given by the matrix $\mathbf{F}$ in (\ref{eq:lowdimf}). This learning task can be formulated as a problem of maximum likelihood estimation. To see this, one can view $\mathbf{G}_{\mathcal{N}}^{m}$ as a realization of a random graph on the node set $\mathcal{N}$ whose joint probability distribution is dictated by the feature matrix $\mathbf{F}$. We calculate an estimator for $\mathbf{F}$ in (\ref{eq:lowdimf}), $\widehat{\mathbf{F}}$, that maximizes the joint likelihood 

\begin{equation}\label{eq:likelihood} \mathbb{L}(\mathbf{F} \mid \mathbf{G}_{\mathcal{N}}^{m}) = \mathbb{P}(\mathbf{G}_{\mathcal{N}}^{m} \mid \mathbf{F}), \end{equation}

\noindent where $\mathbb{P}$ is the joint distribution of a multilayer graph with $m$ layers and unique node set $\mathcal{N}$ given the feature representation $\mathbf{F}$. In general, maximization of (\ref{eq:likelihood}) is computationally intractable. We therefore make two simplifying assumptions about the joint distribution $\mathbb{P}$.
Our assumptions rely upon a suitable definition of a \emph{multilayer neighborhood}. Defining the neighborhood of a node is related to the problem of defining the context of a word in a large document from natural language processing. In static unweighted networks, the neighborhood of the node $u$ is often defined as the collection of nodes that share an edge with $u$. This definition is motivated by the homophily principle \cite{mcpherson2001birds}, which posits that nodes with similar features are highly connected to one another in the network. In many cases this definition of a neighborhood is restrictive. This is particularly true when the observed network is only partially observed or when the edges of the network are generated from some underlying noisy process. We instead define a multilayer neighborhood of node $u$ based on a dynamic process across the network. Our construction is a generalization of the random walk constructions from \cite{grover2016node2vec, perozzi2014deepwalk, tang2015line} and is analogous to the defining of communities via Laplacian dynamics as in \cite{lambiotte2008laplacian, mucha2010community}. To be specific, we define the neighborhood of node $u$ as the collection of vertices that are visited over a random walk on the multilayer network ${\bf G}_{\mathcal{N}}^{m}$. We make this more formal below when describing the neighborhood search procedure of the algorithm.

%
Once the multilayer neighborhood of each node has been defined, we make two simplifying assumptions given the feature matrix $\mathbf{F}$. First, we assume that the joint distribution characterizing ${\bf G}_{\mathcal{N}}^{m}$ is the same as the distribution characterizing the collection of neighborhoods in ${\bf G}_{\mathcal{N}}^{m}$. This assumption is reasonable if we believe that the features $\mathbf{F}$ provide the same information as the multilayer network itself. Second, given the feature matrix $\mathbf{F}$ we assume that the neighborhood of a node $v$ depends only on its own feature representation, $\mathbf{f}_v$ and given this representation is independent of the neighborhoods of other nodes $u \in \mathcal{N}$. These assumptions are the same as those made for the node2vec algorithm for static networks \cite{grover2016node2vec} and are analagous to those made for word2vec, which assumes that the joint probability distribution of a collection of text can be characterized by the distribution of the collection of conditionally independent word contexts given each word's feature representation \cite{mikolov2013efficient}. 

With the conditional independence assumptions of the neighborhoods given $\mathbf{F}$, maximizing the joint likelihood of $\mathbf{F}$ given the {entire} network $\mathbf{G}_{\mathcal{N}}^{m}$ reduces to the task of identifying the features $\mathbf{f}_v$ given the {neighborhood} of $v$ in $\mathbf{G}_{\mathcal{N}}^{m}$. Let $Ne(u)$ denote the neighborhood of node $u$, namely the collection of nodes that are linked to $u$. Given the neighborhood of each node, the likelihood from (\ref{eq:likelihood}) simplifies to

\begin{align}\label{eq:skip_gram}
	\mathbb{L}(\mathbf{F} \mid \mathbf{G}_{\mathcal{N}}^{m}) 
	& = \prod_{u \in \mathcal{N}}\mathbb{P}\left(\text{Ne}(u) \mid \mathbf{f}_u\right).
\end{align}

As it remains a challenging task to quantify the dependence between the neighborhoods of differing layers, the maximization of (\ref{eq:skip_gram}) is still computationally difficult. Thus, we define a family of multilayer graphs for which this maximization is feasible. It turns out that we can define such a family by assuming minimal conditional independence assumptions given the representation $\mathbf{F}$, described as follows. Let $\mathcal{G}_{\mathcal{N}}^m$ denote the family of multilayer graphs whose members are random graphs with $m$ layers and unique nodes $\mathcal{N}$. For every member of $\mathcal{G}_{\mathcal{N}}^m$, assume that the following hold

\begin{itemize}
	\item[(A1)] For all $u \in \mathcal{N}$, $ \mathbb{P}(\text{Ne}(u) \mid \mathbf{f}_u) = \prod_{v \in \text{Ne}(u)} \mathbb{P}(v \mid \mathbf{f}_u) $
	\vskip .5pc
	\item[(A2)] Let $u \in \mathcal{N}$. For every $v \in \text{Ne}(u)$, $\mathbb{P}(v \mid \mathbf{f}_u) = \mathbb{P}(u \mid \mathbf{f}_v).$
\end{itemize}

Assumption (A1) characterizes the local conditional independence among nodes in the neighborhoods of a node $v$ given its feature representation, $\mathbf{f}_v$. Assumption (A2) enforces a symmetric effect of neighboring nodes in their feature space. A consequence of (A2) is that for any node $v$ that is a neighbor of $u$, the following relationship holds

$$\mathbb{P}(v \mid \mathbf{f}_u) = \dfrac{\exp\{\mathbf{f}_v^T \mathbf{f}_u \}}{\displaystyle\sum_{w \in \mathcal{N}} \exp\{ \mathbf{f}_w^T \mathbf{f}_u\}}.$$

 
If the observed graph $\mathbf{G}_{\mathcal{N}}^{m}$ is a realization of a multilayer random graph from the family $\mathcal{G}_{\mathcal{N}}^m$ under which assumptions (A1), and (A2) hold, then (\ref{eq:skip_gram}) can be expressed as

\begin{equation}\label{eq:simplified_likelihood}
	\mathbb{L}(\mathbf{F} \mid \mathbf{G}_{\mathcal{N}}^{m}) = \prod_{u \in \mathcal{N}} \prod_{v \in \text{Ne}(u)} \dfrac{\exp\{\mathbf{f}_v^T \mathbf{f}_u \}}{\displaystyle\sum_{w \in \mathcal{N}} \exp\{\mathbf{f}_w^T\mathbf{f}_u\}}.
\end{equation}

Maximizing (\ref{eq:simplified_likelihood}) is equivalent to maximizing the following log-likelihood

\begin{equation}\label{eq:log_likelihood}
\mathcal{L}(\mathbf{F} \mid \mathbf{G}_{\mathcal{N}}^{m}) = \displaystyle\sum_{u \in \mathcal{N}} \displaystyle\sum_{v \in \text{Ne}(u)} \left[\mathbf{f}_v^T \mathbf{f}_u - \log(Z_{u})\right],
\end{equation}

\noindent where $Z_u = \sum_{w \in \mathcal{N}} \exp\{\mathbf{f}_w^T \mathbf{f}_u\}$ is a normalization constant for the node $u$. Following the approach of \cite{grover2016node2vec, mikolov2013efficient}, we approximate $Z_u$ using negative sampling. We note however that Markov chain Monte Carlo sampling methods could also be used to approximate $Z_u$ as in \cite{wilson2017stochastic, denny2017gergm}. The use of Skip-gram with negative sampling is appealing for two reasons: (i) the algorithm is fast and scalable to large multilayer networks, and (ii) the strategy is closely related to matrix factorization \cite{levy2014neural, qiu2018network} as we will see in Section 5.

Given an observed multilayer network $\mathbf{G}_{\mathcal{N}}^{m}$, multi-node2vec is an approximate algorithm that estimates $\mathbf{F}$ through maximization of the log likelihood function in (\ref{eq:log_likelihood}). The algorithm consists of two key steps. First, the {\bf NeighborhoodSearch} procedure identifies a collection of $s$ neighborhoods of length $l$ for $\mathbf{G}_{\mathcal{N}}^{m}$ through second order random walks on the network. The {\bf NeighborhoodSearch} procedure depends on three hyperparameters - $p$, $q$, and $r$ - that dictate the exploration of the random walk away from the source node and the tendency to traverse layers. Once a collection of neighborhoods or \emph{BagOfNodes} have been identified, the log-likelihood in (\ref{eq:log_likelihood}) is optimized in the {\bf Optimization} step using stochastic gradient descent on the two-layer Skip-gram neural network model of context size $k$. The result of the {\bf Optimization} procedure is a $D$-dimensional feature representation $\mathbf{F}$. Figure \ref{fig:toy} provides an illustration. We describe the {\bf NeighborhoodSearch} and {\bf Optimization} procedures in more detail next.


\begin{figure*}[ht]
	\centering
    \includegraphics[trim = 0cm 0cm 0cm 0cm, clip = TRUE, width = 0.9\linewidth]{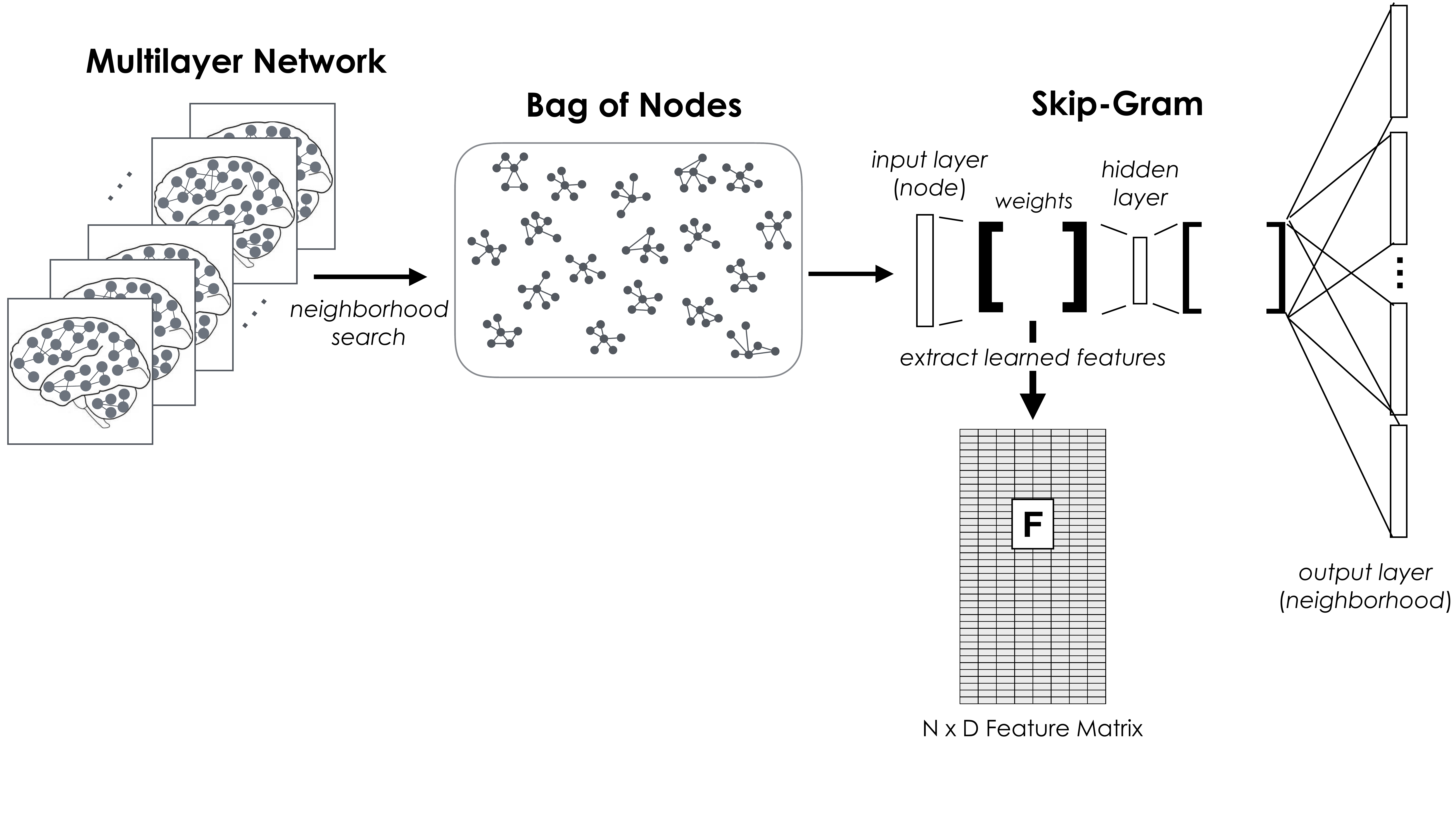}
	\caption{\label{fig:toy} Demonstration of the multi-node2vec algorithm. Beginning with a multilayer network (left), one first identifies a collection of multilayer neighborhoods (Bag of Nodes) via the {\bf NeighborhoodSearch} procedure. Next, the {\bf Optimization} procedure calculates the maximum likelihood estimator $\mathbf{F}$ through the use of the Skip-Gram neural network model (right) on the identified Bag of Nodes.}
\end{figure*}

\subsection{The NeighborhoodSearch Procedure}
Multi-node2vec begins by parsing a multilayer network into a collection of neighborhoods for each unique node in $\mathcal{N}$. The {\bf NeighborhoodSearch} procedure identifies this collection of neighborhoods, or \emph{BagofNodes}, using $s$ truncated second order random walks of length $l$. Without loss of generality, suppose that node labels among layers are registered in the sense that node $u$ in vertex set $V_{\ell}$ represents the same actor as node $u$ in vertex set $V_{\ell'}$. To construct the random walk, we consider the collection of weights $\{w_{\ell, \ell'}(u,v): \ell, \ell' \in 1, \ldots, m; u, v \in \mathcal{N}\}$, where $w_{\ell, \ell'}(u,v)$ defines the edge weight between node $u$ from layer $\ell$ and node $v$ from layer $\ell'$. Thus the collection of edge weights $\{w_{\ell, \ell'}(\cdot, \cdot): \ell \neq \ell'\}$ represent the \emph{inter-layer edges}; whereas, the collection $\{w_{\ell, \ell'}(\cdot, \cdot): \ell = \ell'\}$ represent the \emph{intra-layer edges} in the multilayer network.

For an observed multilayer network and its edge weights defined as above, the {\bf NeighborhoodSearch} procedure identifies $s$ neighborhoods using second order random walks over the nodes and layers of length $\ell$, constructed as follows. Let $u_i$ be the $i$th node visited by the random walk and $\ell_i$ the corresponding layer. Suppose, without loss of generality, that the initial pair $(u_1, \ell_1)$ is chosen uniformly at random. Subsequent vertex, layer pairs are visited according to the conditional probability

\begin{equation} \label{eq:cond_prob}\mathbb{P}(u_i = x, \ell_i = \ell' \mid u_{i-1} = v, \ell_{i - 1} = \ell) = \dfrac{\pi_{v,x,\ell,\ell'}}{Z}, \hskip .5pc w_{\ell, \ell'}(v, x) > 0 \end{equation}

\noindent where $\pi_{v,x,\ell,\ell'}$ is the unnormalized transition probability of moving from vertex-layer pair $(v, \ell)$ to pair $(x, \ell')$, and $Z$ is a normalizing constant. We set $\pi_{v,x,\ell,\ell'}$ as a function of the walk parameters $p, q,$ and $r$ as follows

\begin{equation} \label{eq:transition_probs} \pi_{v,x,\ell,\ell'} = \alpha_{pqr}(t, x, \ell, \ell') \cdot w_{\ell, \ell'}(v, x). \end{equation}

The $\alpha_{pqr}(t, x, \ell, \ell')$ term acts as a search bias on the observed weights that depends on the previously traversed edge $(t,v)$. That is, the walk now resides at node $v$ having just traveled from node $t$ and the next node that the random walk visits depends on (a) the distance $t$ is from the future node, and (b) whether there is a layer transition. Let $d_{\ell}({t, x})$ denote the shortest path distance between nodes $t$ and $x$ in layer $\ell$. To account for layer transitions, we further decompose $\alpha_{pqr}(t, v, x, \ell, \ell')$ as

\begin{equation} \label{eq:linear_function} \alpha_{pqr}(t, x, \ell, \ell') = \beta_{pq}(t, x)~\mathbb{I}(\ell' = \ell) + \gamma_r(v, x)~\mathbb{I}(\ell' \neq \ell), \end{equation}

\noindent where $\beta_{pq}(t,x) = p^{-1} \mathbb{I}(d_\ell(t,x) = 0) + \mathbb{I}(d_\ell(t,x) = 1) + q^{-1} \mathbb{I}(d_\ell(t,x) = 2)$ and $\gamma_r(v, x) = r^{-1} ~ \mathbb{I}(x = v)$. The $\beta_{pq}(t,x)$ term controls the rate at which the random walk explores and leaves the neighborhood of a node within layer $\ell$. This quantity is the same as that specified for static networks in node2vec and has been shown to identify neighborhoods that interpolate between outcomes of breadth first search and depth first search. The return parameter $p$ controls the likelihood of revisiting the same node, layer pair; whereas, the in-out parameter $q$ controls exploration of the walk in layer $\ell$. The $\gamma_r(v, x)$ term controls the rate at which a random walk transitions from one layer to another. Setting the layer walk parameter $r$ to be large ($> max(p, q, 1)$) ensures little layer-to-layer exploration. Setting $r$ in this way encourages independent neighborhood sampling  across layers. On the other hand, setting $r$ to be small ($< min(p,q,1)$) promotes exploration among layers, and the resulting neighborhoods will reflect dependency among the layers. 

Once the parameters $s$, $l$, $p$, $q$, and $r$ have been chosen, $s$ random walks of length $l$ are performed on the nodes of the observed multilayer network using transition probabilities from (\ref{eq:cond_prob}). These $s$ samples serve as the \emph{BagofNodes} from which the nodal features are learned.

\subsection{Optimization}
For a given dimension size $D$, a context size $k$, and the collection of neighborhoods from the {\bf NeighborhoodSearch} step, multi-node2vec then minimizes the cost of (\ref{eq:log_likelihood}) using stochastic gradient descent and the Skip-gram two-layer neural network model. For each node the normalization constant $Z_u$ is approximated using negative sampling. The Skip-gram model iteratively updates the matrix $\mathbf{F}$ in the following manner. Each node is encoded as a one-hot vector and provided as the input layer to a 2-layer neural network from which the neighborhood of the node is predicted. Applying the log-likelihood $\mathcal{L}$ as a cost function, the error of the prediction is calculated. Partial derivatives of the cost function with respect to the rows of each of the intermediate weight matrices are calculated and updated using stochastic gradient descent to minimize cost. This procedure is repeated across all nodes in $\mathcal{N}$ until the cost function can no longer be reduced. After learning from each of the neighborhoods in our bag of nodes, we extract the model’s $node$ $embeddings$ - the $N \times D$ representation weight matrix associated with Skip-gram's input layer. 

This optimization is analogous to that of the node2vec algorithm, but in our application the weight matrices of the two layer neural network are $D$-dimensional representations of the unique nodes $\mathcal{N}$ and thus account for the dependence among layers in the multilayer network. It should be noted that multi-node2vec is an approximate algorithm that relies upon the normalizing constants $\{Z_u\}$, as well as the approximate optimization of stochastic gradient descent. Though not the focus of this paper, there has been a lot of recent work investigating the optimality landscape of gradient descent methods (see for example \cite{lee2016gradient}), which provides promising theoretical justification for its use. 

The choice of $k$ directly affects the amount of information one gains for each node but its value depends on the sparsity of the observed network. Large values of $k$ introduce undesired noise to the identified neighborhoods; whereas, values of $k$ that are too small result in neighborhoods that do not contain significant information about the neighborhoods in the network. We found that setting $k$ near the average degree of the network provided the best results in our numerical studies. In the case that the observed network is either densely connected or contains few layers, the neighborhoods for each node may not contain sufficient information to inform the desired features. In such scenarios, it may be desirable to sample multiple neighborhoods for each node. Thus, we include an optional parameter $a$ that specifies the minimum number of samples generated for each node. Unless otherwise specified, we set $a = 1$ in our numerical studies. Finally, the dimensionality parameter $D$ should be chosen to provide sufficient information about the multilayer network while greatly reducing the total number of nodes $N$, though it is an open problem to understand an optimal dimension to represent general static networks. 

\section{Efficient Implementation of multi-node2vec}

Consider a multilayer functional connectivity network $\mathcal{G}_\mathcal{N}^m$ with $N$ unique nodes and $m$ layers and non-negative edge weights. By construction, multi-node2vec requires the storage of $O(m N^2 + N m^2)$ different edge weights, since there are $O(m N^2)$ intra-layer edges and $O(N m^2)$ inter-layer edges. This can quickly overwhelm computational resources when the number of layers or unique nodes is large. It turns out that multi-node2vec can be applied by only storing $O(N^2)$ values, which greatly improves the efficiency of the algorithm. We begin by analyzing the relationship of multi-node2vec with the node2vec and DeepWalk algorithms. To do so, we need a notion of equivalence between two stochastic algorithms. For this purpose, we consider the stochastic equivalence of two algorithms, defined as follows.

\begin{definition}
	Let $A_1$ and $A_2$ be two stochastic algorithms, each with the same set of possible outcomes $\Omega$. That is, for fixed input data $X$, $A_k$ is a random function that maps $X$ to an outcome $o \in \Omega$: $A_k(X) \rightarrow o \in \Omega$. Define $\mathbb{P}_k$ as the probability mass function characterizing the probability of each possible outcome of $A_k$:
	$\{\mathbb{P}_k(A_k(X) = o): o \in \Omega\}$. $A_1$ and $A_2$ are said to be {\bf stochastically equivalent} if $\mathbb{P}_1 = \mathbb{P}_2$.
\end{definition}

Let $\mathbb{A}$ denote the $N \times N$ aggregate adjacency matrix of the nodes $\mathcal{N}$ with entries $\mathbb{A}_{u,v} = \sum_{\ell}\sum_{\ell'} w_{\ell, \ell'}(u,v)$. Define the adjusted version of $\mathbb{A}$, $\widetilde{\mathbb{A}}(r)$, as the $N \times N$ matrix with entries

$$\widetilde{\mathbb{A}}_{u,v}(r) = r^{-1}\sum_{\ell \neq \ell'}w_{\ell, \ell'}(u,v) + \sum_{\ell}w_{\ell, \ell}(u,v), \hskip 1pc u,v \in \mathcal{N}.$$

Note that $\widetilde{\mathbb{A}}(r) = \mathbb{A}$ when $r = 1$. One can view the matrix $\widetilde{\mathbb{A}}(r)$ as an adjusted adjacency matrix whose edge weights depend on the layer walk parameter $r$. Write $\widetilde{\mathbf{G}}_{\mathcal{N}}{(r)}$ as the graph with nodes $\mathcal{N}$ and edge weights specified by the adjacency matrix $\widetilde{\mathbb{A}}(r)$.

The following lemma relates multi-node2vec with node2vec and DeepWalk and shows under what conditions they are stochastically equivalent in terms of the walk parameters $p, q,$ and $r$. 

\begin{lemma}\label{thm:equiv}
	Let ${\mathbf{G}}_{\mathcal{N}}^m$ be an observed multilayer network and let $\widetilde{\mathbf{G}}_{\mathcal{N}}{(r)}$ be its adjusted aggregate network. Suppose that the parameters $D, k, s, l$ are held constant. Then the following hold
	\begin{enumeratea}
		\item For all $p, q, r > 0$, the application of multi-node2vec to $\mathbf{G}_{\mathcal{N}}^m$ is stochastically equivalent to the application of node2vec to $\widetilde{\mathbf{G}}_{\mathcal{N}}{(r)}$.
		\item If $p = q = 1$, the application of multi-node2vec to $\mathbf{G}_{\mathcal{N}}^m$ is stochastically equivalent to the application of DeepWalk to $\widetilde{\mathbf{G}}_{\mathcal{N}}{(r)}$.
	\end{enumeratea}
\end{lemma}

\begin{proof} Since multi-node2vec, node2vec, and DeepWalk all use Skip-Gram on identified neighborhoods, it will suffice to show that the transition probabilities of the random walks used to identify the neighborhoods for each method are equal under the stated conditions to prove Theorem \ref{thm:equiv}. We begin by proving part (a) for general $p,q,r > 0$. Let $\pi_{u,v}$ denote the unnormalized transition probability of the random walk traveling from $u \rightarrow v$ based on the application of node2vec on the graph $\widetilde{\mathbf{G}}_N(r)$. Similarly let $\pi^*_{u,v}$ denote this unnormalized transition probability of the random walk based on the application of multi-node2vec to $\mathbf{G}_{\mathcal{N}}^m$. Then by the law of total probability we have
	\begin{align*}
	 		\pi^*_{u,v}: &= Z \cdot P_{\mathcal{G}_{\mathcal{N}}^m}(u_{j+1} = u \mid u_{j} = v)
	 		= Z \cdot \sum_\ell \sum_{\ell'} w_{\ell, \ell'}(v, x) P(\ell_{i-1} = \ell)\\
			& = \beta_{pq}(t,v) \sum_{\ell} w_{\ell, \ell}(u, v) + \gamma_r(u, x) \sum_{\ell \neq \ell'} w_{\ell, \ell'}(u,v). 
		\end{align*}

Note that $\beta_{pq}(t,v) = 1$ when $v = u$ and that $\gamma_r(u,x) = r^{-1}$. It follows that $\pi^*_{u,v} = \pi_{u,v}$ and thus part (a) is proved. Part (b) is proven in an analogous fashion by taking $\pi_{u,v}$ as the transition probability for the random walk associated with DeepWalk on the graph $\widetilde{\mathbf{G}}_N(r)$ and noting that $\beta_{pq}(t,v) \equiv 1$ when $p = q = 1$.
\end{proof}

Lemma \ref{thm:equiv} reveals that the application of multi-node2vec on an observed multilayer network $\mathbf{G}_{\mathcal{N}}^m$ is stochastically equivalent to the application of node2vec on the adjusted aggregate graph $\widetilde{\mathbf{G}}_{\mathcal{N}}{(r)}$. In practice, this means that running multi-node2vec on an observed multilayer network will provide the same results as running node2vec on the corresponding adjusted aggregate network if the same seed set is specified for a random number generator. This suggests that multi-node2vec can be implemented with just the storage of $\widetilde{A}(r)$, which contains $O(N^2)$ edge weights. In the special case that $p = q = 1$, one can equivalently run multi-node2vec, node2vec, or DeepWalk.\\ 


\section{Numerical Study}

We now apply multi-node2vec to a multilayer brain network representing the functional connectivity of 74 healthy individuals and 60 patients with schizophrenia who underwent resting state fMRI. In this case study, we demonstrate the use of multi-node2vec for three primary objectives: (i) clustering of brain regions into communities of similar features, (ii) classification of nodes into anatomical regions of interest in the brain, and (iii) comparing and classifying two populations of individuals. To assess overall performance, we compared multi-node2vec with several off-the-shelf embedding techniques, including LINE, DeepWalk, and node2vec. Our analysis reveals that multi-node2vec identifies features that closely associate with the functional organization of the brain and provides a powerful strategy for comparing across groups of individuals. 

To analyze the efficacy of multi-node2vec, we consider the tasks of clustering and classification of ROIs using the subnetwork labels as ground truth. We furthermore analyze multi-node2vec via a classification study, where we aim to classify healthy individuals from schizophrenia patients using global summaries of the identified multilayer embeddings. Publicly available code for the multi-node2vec algorithm as well as all code used for our findings are available at \url{https://github.com/jdwilson4/multi-node2vec}.

\subsection{Description of Data}
\vskip .01pc
We investigate a data set of resting-state fMRI scans of 74 healthy individuals (ages 18-65, 23 female) and 60 individuals with Schizophrenia (ages 18-65, 16 female) from the Center for Biomedical Research Excellence (COBRE \cite{mayer2013functional}) posted to the 1000 Functional Connectomes Project \cite{biswal2010toward}. Participants had no history of neurological disorder, mental retardation, substance abuse or dependencies in the last 12 months, or severe head trauma. Participants underwent 5 minutes of resting state fMRI in which they had no task except to stay awake, followed by a multi-echo MPRAGE scan (see \cite{mayer2013functional} for scanning parameters and preprocessing information). 

To construct the multilayer representation of this data set, we use a previously validated atlas \cite{power2011functional} that specifies 264 spheres of radius 8mm, which constitute our 264 regions of interest (ROIs). We averaged the fMRI time-series' from all voxels within each ROI, yielding 264 time-series' per participant. For each of these time series', we regressed out 6 motion parameters (to account for head movement), 4 parameters corresponding to cerebrospinal fluid, and 4 parameters corresponding to white matter. These steps have been shown to reduce bias and noise within the data \cite{chai2012anticorrelations}. Finally, for each participant, we correlated the 264 time-series' with one another, yielding a 264 * 264 correlation matrix for each participant. We analyze the weighted multilayer network representation of these data. Intra-layer edges are encoded with a weight of $w_{\ell, \ell}(u,v)\mathbb{I}(r(u,v) > 0)$, where $r(u,v)$ is the correlation between the two incident regions. Inter-layer edges are encoded as $w_{\ell, \ell'}(u,v) = 1$ when $u = v$ and 0 otherwise. The ground-truth subnetwork labels are previously defined functional subnetworks, established in \cite{power2011functional}. The subnetwork labels and number of regions belonging to each functional subnetwork is presented in Table \ref{tab:summary}.

\vskip 1pc
\begin{table}[ht]
\caption{The number of regions in each functional subnetwork of the whole brain when the Power atlas parcellation is applied to each whole brain network. These subnetworks are used to demonstrate the use of multi-node2vec in the classification study in Section 4.3. \label{tab:summary}}
\centering
\begin{tabular}{l | l | l}
~Auditory: 13 ~& ~Dorsal Attention: 11 ~& ~Default Mode: 58 \\ 
~Salience: 15 ~&~Memory/retrieval: 5 ~& ~Fronto-parietal Task Control: 25 \\
~Visual: 31~ & ~Ventral Attention: 9~ &~Subcortical: 13~ \\ 
~Cerebellar: 4~ & ~Uncertain: 28~ & ~Cingulo-opercular Task Control: 14 \\
~Sensory -- Hand: 30~ & ~Sensory -- Mouth: 5 & \\
\end{tabular}
\end{table}

\subsection{Simulation Study and Parameter Choices}
Before applying multi-node2vec to the fMRI data, we first describe a strategy for choosing the parameters for the algorithm through the use of simulation and theoretical study. In \cite{grover2016node2vec}, the authors recommended values for the walk parameters $p$ and $q$ based on extensive empirical studies for node2vec. For consistency and ease of comparison, we use their recommended values ($p = 1, q = 0.5$) for both node2vec and multi-node2vec in our application though we mention that these too could be tuned using simulation. In Section 5, we study the limiting behavior of multi-node2vec as a function of the walk length parameter, $l$, and find that asymptotically in $l$ the embeddings from multi-node2vec converge to the result of non-negative matrix factorization. To balance computational speed and theoretical gaurantees, we therefore suggest using a moderately sized $l$ in application and opt for $l = 30$. For the layer walk parameter $r$, we assess the performance of three different values, $r = 0.25, .5, .75$, on the fMRI dataset to investigate differences among these values.

Through simulation, we are particularly interested in three aspects of multi-node2vec: (i) analyzing the specificity of multi-node2vec, (ii) investigating the effects of the neighborhood or context size of the identified neighborhoods, $k$, and the dimension of the feature vectors $D$ as they relate to the size, structure, and connectivity of the network, and (iii) analyzing the scalability of multi-node2vec for networks with a large number of nodes and/or layers.  

For (i), we simulated a multilayer graph where each layer was an independent \erdos random graph with probability of connection set to the average degree of that group. These simulated graphs represent what a multilayer network would look like at random with no topological structure other than preserving the average degree of the group of images. 

For (ii) and (iii), we use unweighted multilayer networks using a multilayer generalization of the planted partition model, designed to align with the connectivity, clustering, and size of the observed multilayer networks in our fMRI study. In all simulations, each layer of the simulated multilayer network contains $n = 264$ nodes to match the fMRI networks in our application. Nodes were placed randomly into $c$ equally-sized communities. For each layer, edges are placed randomly between two nodes of the same community with probability $p_{in}$ and edges are placed between two nodes of differing communities with probability $p_{out}$. With this construction, each layer of the generated network has the same community structure across layers. This graph model is a special case of the multilayer stochastic block model (MSBM) considered \cite{han2015consistent, stanley2015clustering, wilson2016community}. For our analysis, nodes of the same community are expected to have similar features with one another and different features than nodes from other communities. This model therefore provides a well-structured multilayer network for which we can study and tune multi-node2vec. We analyze the effect of $k$ on the performance of the algorithm on the MSBM. We note that one could also tune the dimension parameter $D$ through an analogous simulation study. In our application, we have access to ground-truth labels for the nodes -- the functional subnetwork label -- and therefore directly compare the performance of multi-node2vec against competing methods by running each method across a grid of dimension $D$ ranging from 2 to 100.

For each of the following studies, we set $p_{in} = 0.49$ to match the average degree of the functional brain networks. To assess the relevance of the features identified by multi-node2vec, we compare the clusters obtained from the k-means algorithm on the feature matrix with the true community labels of the network and calculate the adjusted rand score as a measurement of match between the two partitions. For each simulation, we replicate the study 30 times and report the average adjusted rand score. The results for each simulation is presented in Figure \ref{fig:sim3} and discussed below.


\begin{figure}[ht]
\centering
\includegraphics[trim = 0cm 0cm 0cm 0cm, width = 0.75\linewidth]{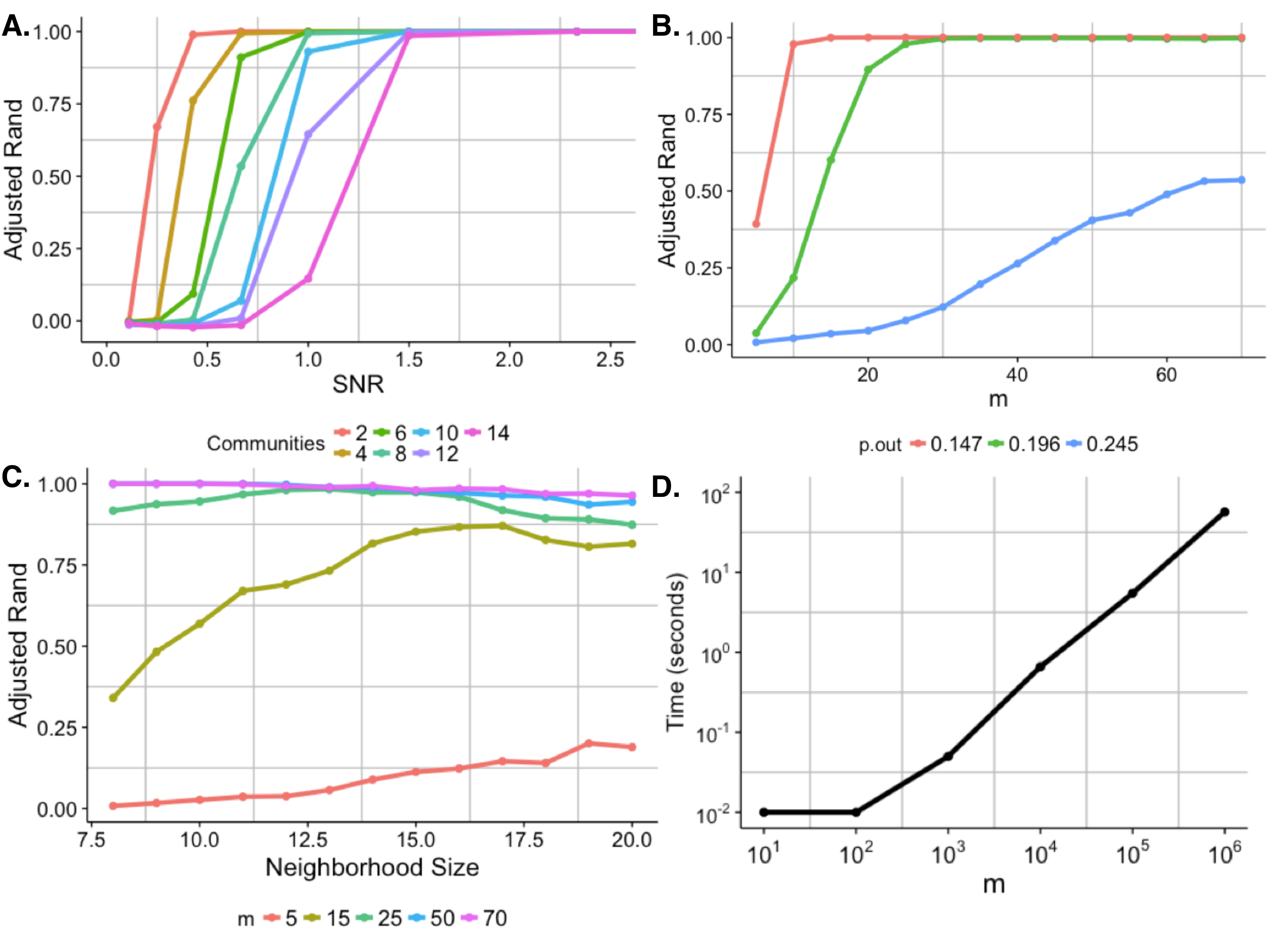}
\caption{Simulation results from the numerical study described in Section 6. All simulations are repeated 30 times and the average is shown. {\bf A.} The adjusted rand index score of the clusters identified by k-means clustering on the identified feature matrix from multi-node2vec applied to the multilayer stochastic block model as a function of the signal to noise ratio: SNR = $p_{in} / p_{out} - 1$ and {\bf B.} across the number of layers in the network. {\bf C.} The adjusted rand index score of the clusters identified by k-means clustering on the identified feature matrix from multi-node2vec as a function of the neighborhood size input to the algorithm. {\bf D.} The average time (in seconds) required by multi-node2vec on multilayer random graphs with 10 nodes in each layer and $m$ layers. Notably, networks with 1 million layers required just 58 seconds. \label{fig:sim3}}
\end{figure}

\subsubsection{Specificity of multi-node2vec}
To test the specificity of the results identified in our study, we first applied multi-node2vec to multilayer \erdos random graphs of the same size and expected degree as the populations that we investigated. We expect that the embeddings of completely random multilayer graphs would give no structural insights, and thus that the clusters identified from the embeddings would not closely align with the functional subnetworks. This test provides a validation that multi-node2vec can effectively distinguish real signal from noisy networks. To test this, we first identified the embeddings on a simulated multilayer network using multi-node2vec. Then, we identified 13 clusters from the embeddings and calculated the adjusted rand index (as done in our application study) of the clusters with the true subnetwork labels. We repeated this across 100 simulated multilayer networks from both the healthy and patient groups. 

For the simulated networks for the healthy group, we calculated an average adjusted rand of 0.25 (st. deviation 0.08). For the simulated networks representing the patient group, we calculated an average adjusted rand of 0.21 (st. dev 0.10). These results reveal that there is no structure in the embeddings of these random multilayer graph models for each group, and suggest that the multi-node2vec algorithm does not incorrectly identify structure in a noisy network.

\subsubsection{Sensitivity of multi-node2vec}
\noindent \underline{Community Strength}\\
We first investigate the effects of the strength of community structure on multi-node2vec. To do so, we varied the out-group probabilities $p_{out}$ to be between 10\% - 90\% of $p_{in}$ and assess the performance of the algorithm over values of the signal to noise ratio (SNR) = $p_{in}/p_{out} - 1$. We simulated multilayer networks like this with 74 layers, across $c = 2$ to $14$ communities per layer. Results are shown in plot {\bf A.} of Figure \ref{fig:sim3}. We observe that as the disparity between in-group and out-group probabilities increased, the feature embeddings more clearly represented the community structure in the graph. Furthermore, across all values of $p_{out}$ the performance of multi-node2vec improved as the number of communities decreased. For multilayer networks with 2 communities, the feature embeddings perfectly represented the community structure for values of SNR greater than or equal to 0.4. Networks with 14 communities per layer required SNR greater than 2.0 to achieve the same result. These results provide evidence that the feature embeddings identified by multi-node2vec are able to efficiently capture the community structure of multilayer networks. 


\noindent \underline{Effect of the Number of Layers}\\
We next analyze the effect of the number of layers on the multi-node2vec algorithm. In this simulation, we generated multilayer graphs from the planted partition model with $m$ = {5, 10, 15, … 65, 74}. As before, we fixed $p_{in} = 0.49$ and varied $p_{out} = 0.245, 0.196,$ and $0.147$ to match the best three values from the community strength simulations. We report the average adjusted rand from 30 replications on networks with $c = 12$ communities in plot {\bf B.} of Figure \ref{fig:sim3}. For all three values of $p_{out}$, the performance of multi-node2vec consistently improves across an increasing number layers. This result supports the belief that each layer provides additional neighborhood information for each node from which the multi-node2vec algorithm can efficiently learn. 



\noindent \underline{Effects of Context Size, $k$}\\
To test the effect of neighborhood size, $k$, we ran simulations of the planted partition model multilayer networks with $m$ = {5, 15, 25, 50, and 74}  over a range of 8 - 20 nodes per neighborhood with $p_{out} = 0.245$. We plot the average adjusted rand of the clusters identified on the feature matrix for networks with $c = 12$ communities in the plot {\bf C.} of Figure \ref{fig:sim3}. We find that the algorithm improves with an increasing context size; however, the number of layers in the network has more impact on the performance of the algorithm. Indeed, when $m \geq 25$, the neighborhood size does not significantly affect (if at all) the performance of the algorithm. On the other hand, for a small number of layers (say, $m = 5$) the increasing the context size plays a more important role in its identified features. Thus, for multilayer networks with a large enough of layers, the context size will not dramatically affect the results of multi-node2vec, but in networks with fewer than 25 layers, one should carefully tune this parameter.


\noindent \underline{Scalability}\\
Identifying a neighborhood for the bag of nodes needed for the algorithm relies upon a random walk strategy, which can be done in constant time using alias sampling (as done in the node2vec algorithm). The optimization part of the algorithm turns out to be linear in the number of distinct nodes in the multilayer network. Notably, this is drastically faster than the spectral decomposition of the network, which in the best case scenario is of cubic in the unique number of nodes. To show this empirically, we consider multilayer networks with $n = 10$ unique nodes in each of $m$ total layers. We apply multi-node2vec on planted partition networks across a range the number of layers $m$ from 10 to 1 million layers. We calculate the amount of time (in seconds) required for multi-node2vec with fixed $k = D = 5$ on 30 replications and report the average time in the plot {\bf D.} of Figure \ref{fig:sim3}. For networks with 1 million layers, multi-node2vec took on average only 58 seconds. We note that the complexity of multi-node2vec as a function of $n$ is also linear, and this is justified with the scalability analysis in \cite{grover2016node2vec}. This figure suggests that the multi-node2vec algorithm is linear in the number of layers in the network, and provides evidence that this algorithm is well-suited for embedding massive multilayer networks. 

\subsection{Analysis of Schizophrenia Data}

Based on our discussion and results in Sections 4.1 and 4.2, we set $k = 10$, and $l = 30$. We set $p = 1$, and $q = 0.5$ to match the parameter settings of node2vec as suggested in \cite{grover2016node2vec}, and we investigated the effects of the layer walk parameter $r = 0.25$, 0.50, and 0.75. 

\subsubsection{Clustering Regions of Interest}

To explore functional region segmentation, we first clustered the rows of the feature matrices identified from multi-node2vec across all three walk parameter settings. For this task, we were particularly interested in the effect of the feature dimension on clustering performance. To test this effect, we proceeded as follows. Multi-node2vec was run using $D$ features. The k-means clustering algorithm was then applied on the rows of the resulting $N \times D$ matrix, and the number of clusters was set to 13 to match the true number of subnetwork labels. For each run, the identified clusters were compared against the true subnetwork labels using the adjusted rand score. We repeated this process for each method across a grid of $D$ from 2 to 100 in increments of 2.

The match of the identified clusters with the ground truth improves as the number of features, $D$ increases. Notably, even for $D$ as small as 6, the ROI clusters closely resemble the ground truth labels (adjusted rand $\approx$ 0.83). We note that such clustering analyses provide a heuristic for assessing how many dimensions should be used to capture a desired ground truth in a multilayer network. For example, in this case we can use even just 2 dimensions and still capture more than 80$\%$ of the functional organization of the healthy individuals. These results reveal that the features of multi-node2vec provide practically relevant information about the functional subnetwork to which these ROIs belong. This finding is further supported in the classification study performed next. 


\subsubsection{Classification of Functional Subnetworks}

We now assess the utility of the features learned from multi-node2vec through the classification task of predicting the functional subnetwork location for each ROI in the healthy individuals. We considered the classification of the nine subnetworks containing ten or more ROIs, which included the \emph{auditory, cingulo-opercular task control, default mode, fronto-parietal task control, salience, sensory/somatomotor -- hand, subcortical, visual}, and \emph{dorsal attention} subnetworks. In the classification task, we tested two scenarios for network embedding methods -- (i) the multilayer network representing the resting state fMRI of 74 healthy individuals alone, and (ii) the multilayer network with additional noisy layers.

For each subnetwork, we trained a one-versus-all logistic regression classifier on the rows of the feature matrix for each method on 80$\%$ of the regions using $D$ identified features. We applied the classifier to the remaining 20$\%$ of the ROIs and assessed the performance of the classifier using the area under the curve (AUC). We performed this classification on the feature matrices for each method and calculated the resulting AUC of the classifier across $D$ ranging from 2 to 100 in increments of 2. 

We compared multi-node2vec to several off-the-shelf embedding methods including node2vec, DeepWalk, and LINE. As these methods are single-layer methods, we ran them on the average weighted network of each population where layers were the same as those used for multi-node2vec. For node2vec, we set the return parameter as p = 1 and the in-out parameter as q = 0.5 to guide the neighborhood search following the suggestions of the original paper. For DeepWalk, we kept default parameters. Matching multi-node2vec, we set k = 10 for both node2vec and DeepWalk. For LINE, we used its default parameters: negative-sampling = 5 and $\rho$ = 0.025. To match LINE's default of 1 million training samples, we sampled s = 3,788 neighborhoods for each node in node2vec and DeepWalk. We ran all methods to learn $D$ features, from $D = 2, \ldots, 100.$ All experiments were performed on an AWS T2.Xlarge instance (specs: a 64-bit Linux platform with 16 GiB memory). We report the AUC for each method and each subnetwork when 20 layers of noise were added in Figure  \ref{fig:classification20}. Results for the non-noisy setting and the setting with 10 layers of noise are shown in the Appendix.

Our study reveals that even in the presence of noise, multilayer embeddings of the healthy individuals closely match the functional organization of the brain. Furthermore, multi-node2vec is comparable to the competing methods in the non-noisy setting, where we expect layers to be homogeneous across the healthy patients. We further find that multi-node2vec is robust to multilayer networks with additional noisy layers. Indeed in this setting, we find that multi-node2vec outperforms its competitors in seven of nine classification studies. These results provide evidence of the robustness of multi-node2vec across multilayer networks with heterogeneous layers and reveal the overall utility of the algorithm for noisy and non-noisy networks.


We begin by analyzing the classification result on the original 74 individuals (figure shown in the Appendix.). Since each individual in the original study is healthy, we expect the networks of each these individuals to share similar structure. It follows that the aggregate network provides an unbiased summary of the multilayer network with less variability than each layer alone. Thus methods applied to the aggregate network are expected to do better than multi-node2vec. Despite this, we find that multi-node2vec is comparable to the competing methods for seven out of nine subnetworks and outperforms other methods for small $D$ in the \emph{visual} and \emph{sensory-motor (hand)} regions. The LINE method does particularly well in the \emph{salience} and \emph{dorsal attention} classifications, and outperforms multi-node2vec and all other methods across $D$. All methods improve with increasing $D$ and approach 1, indicating perfect classification.

To test the performance of multi-node2vec on multilayer networks with noise, we next generated $b$ layers, each with 264 nodes to match the number of regions in every other layer, from an \erdos with edge probability set to the average edge density across all 74 layers. In this way, we add $b$ layers of randomly connected nodes that act as noise against the structure present in the 74 individuals in the study. We set $b = 10$ and 20 and re-ran all of the methods with the same parameter settings as in the original study. 

\begin{figure}[ht]
\centering
\includegraphics[width = 0.8\textwidth, angle = 270]{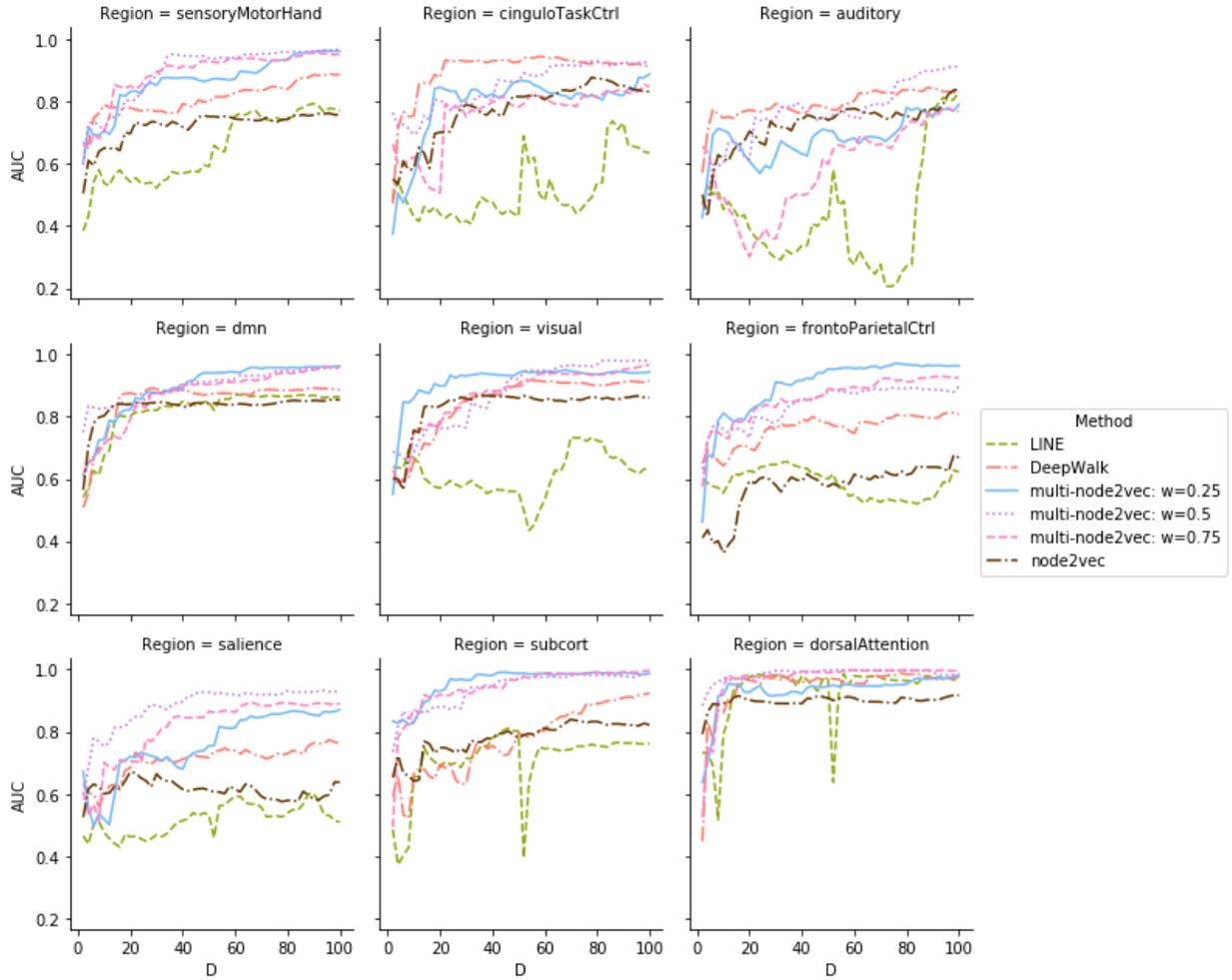}
\caption{The AUC of a one vs. all logistic regression classifier for the nine major functional subnetworks of the brain all 74 healthy individuals and 20 layers of noise. Plots show the AUC of the classifier against the number of dimensions $D$ for feature representations from multi-node2vec, node2vec, DeepWalk, and LINE. \label{fig:classification20}}
\end{figure}

As can be seen in Figure \ref{fig:classification20}, single-layer embedding methods are dramatically affected by the addition of noisy layers; whereas, multi-node2vec is robust to noise. For both $b = 10$ (in the Appendix Figure 4) and $b = 20$ (Appendix Figure 5), all three runs of multi-node2vec outperforms competing methods for seven out of nine of the classification studies. In particular, multi-node2vec has clear advantages over the competing methods in the \emph{subcortical}, \emph{salience}, \emph{sensory-motor (hand)}, and \emph{fronto-parietal task control} regions. Importantly, multi-node2vec's performance is not strongly affected by the addition of more noisy layers suggesting that the features identified by the method align with the true 74 layers of the population. We find that the LINE method is most affected by noise, followed by node2vec. 
%

These results, in combination to the clustering results from the previous section, provide strong evidence that the features engineered from multi-node2vec provide biologically relevant information about the functional organization of the brain, and is generally robust to moderate amounts of noisy layers.

\subsubsection{Comparison of Healthy controls and Patients}
To compare the populations of patients with schizophrenia to their healthy peers, we apply multi-node2vec with $r = 0.25$ to both groups using the same parameters as described above providing a 264 $\times$ 100 network embedding for each group. Importantly, the two embeddings are not directly comparable as features for each population may differ or be arranged in differing order. To assess differences between groups then, one must compare within population summaries across populations. For our study, we compare the variability of the embeddings within each functional subnetwork. To make this precise, let $A$ represent the index of the regions that are contained within a specified functional subnetwork $\mathcal{A}$. Let $f_{g,i}$ denote the $i$th feature vector in group $g$ and $\overline{f}_{g,\mathcal{A}}$ denote the mean vector of the embeddings from region $\mathcal{A}$ in group $g$, where $g = $healthy, patient. Let $\mid \mid x \mid\mid_F$ notate the Frobenius norm of the vector $x$. For each group, we calculate the mean squared deviation for every region $\mathcal{A}$:

$$MSD_{g, \mathcal{A}} = \dfrac{1}{|\mathcal{A}|}\sum_{i \in A}\mid\mid f_{g,i} - \overline{f}_{g,\mathcal{A}}\mid\mid_F^2,$$

\noindent where $|\cdot|$ represents the cardinality of a set. The value of $MSD_{g, \mathcal{A}}$ quantifies the inner regional variability of the embeddings for region $\mathcal{A}$ in the $g$th sample. Large values of $MSD_{g, \mathcal{A}}$ suggest low similarity of nodes within the same region $\mathcal{A}$ and hence higher entropy among that region. For each region $\mathcal{A}$ mentioned in Table \ref{tab:summary}, we compare the mean squared deviation across populations using a two sided t-test on the quantity 

$$MSD_{healthy, \mathcal{A}} - MSD_{patient, \mathcal{A}}.$$

These results are reported in Table \ref{tab:mean_diff}. We find a significant difference in the mean squared deviation in the Default Mode Network (p-value $<$ 0.001) as well as a strong trend within the salience network (p-value = 0.085). In both subregions, the mean squared deviation was found to be lower in the healthy group than in the patient population, suggesting higher variability in the patient group in these two regions. Our findings are well-supported by the Triple Network Model (TNM) theory of the brain \cite{menon2010saliency, seeley2007dissociable}. The TNM explains how individuals switch between externally motivated cognitive processes (i.e., goal directed tasks) that are associated with the central executive networks and internally motivated cognitive processes (i.e., rumination and mind-wandering) that are associated with DMN via the salience network \cite{menon2010saliency, seeley2007dissociable}.

The TNM foremost relates to schizophrenia because of differences observed in the Default Mode Network (DMN) in individuals with schizophrenia. Previous research has indicated increased activity and within network connectivity in the DMN \cite{whitfield2009hyperactivity} and decreased segregation between the DMN and central executive networks \cite{woodward2011functional} in patients with schizophrenia versus healthy individuals. The TNM further posits that pathological salience (inappropriate monitoring by the salience network) may be associated with DMN pathology and consequently many of the symptoms of schizophrenia \cite{menon2011large}. This theory is consistent with recent evidence indicating TNM, and particularly salience network, dysregulation is correlated with symptom severity in patients with schizophrenia \cite{hare2018salience, supekar2019dysregulated}. Our findings that the DMN has significantly smaller variability within healthy individuals than in individuals with schizophrenia as well as the fact that the salience network is statistically different between individuals with schizophrenia and healthy controls empirically supports these findings.

Finally, our results are consistent with a recent meta-analysis investigating the effect of schizophrenia on connectivity \cite{li2019dysconnectivity}, which found consistent hypoconnectivity amongst the DMN in patients with schizophrenia. Notably, however, this meta-analysis also found aberrant connections in several other functional networks, a finding we do not replicate here. Future research is clearly needed to know whether our lack of significant findings in other networks (e.g., auditory, somatomotor) reflects lower power in our study compared to the meta-analysis, or a systematic difference as a result of the vastly different methodological approaches. Given that our results with the Default Mode Network and the Salience Network are consistent with both the meta-analysis as well as other papers using this same dataset (e.g., \cite{wang2014disruptive}), we suspect this is primarily an issue of power, but future work is clearly necessary to fully understand these discrepancies.

\begin{table}[ht]
\caption{Two sided ninety-five percent confidence intervals for the difference of mean squared deviation in healthy controls and schizophrenia patients. Deviations were calculated using 100 features from the network embedding for each group. $^{**}$the sum of squares deviation in healthy controls was less than that in the patients with schizophrenia at a 0.001 level; $^*$the sum of squares deviation in healthy controls was less than that in the Schizophrenia patients at a 0.10 level. \label{tab:mean_diff}}
\vskip .5pc
\centering
\resizebox{\linewidth}{!}{ 
\begin{tabular}{l | c | l | c }
{\bf Subnetwork} & ~{$MSD_{healthy} - MSD_{patient}$} & ~{\bf Subnetwork} & ~{$MSD_{healthy} - MSD_{patient}$}\\
\hline
Auditory & (-0.108, 0.332) & ~Dorsal Attention & (-0.326, 0.106)\\
C-O Task Control & (-0.319, 0.155) & ~Default Mode & ~~(-0.241, -0.031)$^{**}$\\
Salience & ~(-0.373, 0.047)$^*$ & ~Memory/retrieval & (-0.359, 0.581) \\
F-P Task Control & (-0.213, 0.187) & ~Visual & (-0.160, 0.095) \\
Sensory -- Hand & (-0.295, 0.089) & ~Ventral Attention & (-0.204, 0.434) \\
Subcortical & (-0.198, 0.488) & ~Cerebellar & (-0.578, 0.379) \\
Sensory -- Mouth & (-0.440, 0.269) & & \\
\end{tabular}}
\end{table}

\subsubsection{Classification of Patients and Healthy Controls}

We next consider the classification task of differentiating schizophrenia patients from healthy controls using the embeddings from multi-node2vec. We first apply multi-node2vec to each of the 134 total individuals in the study separately (74 healthy and 60 patients) and extract $D = 100$ feature embeddings describing each person's functional connectivity. From these embeddings, we then calculate the mean squared deviance $MSD_{j, \mathcal{A}}$ for each individual $j = 1, \ldots, 134$ and each region $\mathcal{A}$. Using the binary response vector $y = (y_1, \ldots, y_{134})$ where $y_j = 1$ if individual $j$ has schizophrenia and 0 if individual $j$ is a healthy control, we apply several off-the-shelf binary classification techniques -- including k nearest neighbors, logistic regression, an L2 penalized logistic regression and a random forest classifier -- using the mean squared deviance vectors to predict whether or not the individual has schizophrenia. We perform ten fold cross validation and report the average and standard error of the results in Table \ref{table:classification_results}. For k nearest neighbors, we look across a grid of $k$ between 1 and 30 and report the result with the highest mean accuracy.
\begin{table}[ht]
	\caption{Ten fold cross validation results for classification of patitents and healthy controls using individual embeddings. The average and standard error (s.e.) are reported.\label{table:classification_results}}
	\vskip .5pc
	\centering
	\begin{tabular}{l | l}
		\underline{\bf Method} & \underline{\bf ~Mean Accuracy (s.e.)}\\
		k Nearest Neighbors & ~0.718 (0.017)\\
		Logistic Regression & ~0.758 (0.075)\\
		L2 Penalized Logistic ~ & ~0.592 (0.021)\\
		Random Forest & ~0.787 (0.077)\\
	\end{tabular}
\end{table}

Th random forest classifier performs better than the other off-the-shelf methods using our discovered embeddings, and obtains a classification accuracy of 0.787 on average. It is important to reiterate that multi-node2vec is an \emph{unsupervised method}, namely the algorithm is not trained to explicitly distinguish between two populations as is done formally in the network classification problem. With that in mind, there have studies on the COBRE data set that were \emph{supervised} and though these studies are not directly comparable with our result, their comparison does deserve some discussion. 

We compare our findings with the recent work in \cite{relion2019network}, which establishes the highest performance to date on the COBRE data set using supervised edge-based techniques (see Table 1 for their results). Their method achieved an accuracy of 0.927. Furthermore, other edge-based methods that employ variable selection on the edges in each network obtain accuracies on average approximately 0.85. As expected, such supervised methods do indeed outperform our unsupervised strategy. Perhaps the most fair comparison among the results in \cite{relion2019network} to our own result is the comparison of multi-node2vec with network summaries. Like multi-node2vec, network summaries provide a dimension reduction to the original networks and are not explicitly designed for classification. We find that classification via embeddings of multi-node2vec significantly outperform the classification using network summaries, which obtained 0.614 accuracy on average. This study reinforces the fact that multi-node2vec provides biologically relevant information for classification of disease type. In future work, we will investigate developing supervised embedding methods designed specifically to classify disease and other clinical features.
%



\section{Limiting Behavior of multi-node2vec}
Multi-node2vec is an approximate algorithm that seeks to maximize the log-likelihood objective function given in equation (\ref{eq:log_likelihood}). Approximation is needed for two objectives - (i) the identification of multilayer neighborhoods via random walks, and (ii) the application of the Skip-gram neural network model with negative sampling. By analyzing the asymptotic nature of the random walks in the {\bf NeighborhoodSearch} procedure as $l \rightarrow \infty$, one can leverage the recent work on the Skip-gram with negative sampling from \cite{levy2014neural, qiu2018network} to show that multi-node2vec approximates implicit matrix factorization. We describe this main result below.

Denote $\mathcal{D}$ as the collection of neighborhoods identified by the {\bf NeighborhoodSearch} procedure. Let $w = \{u_1, \ldots, u_l\} \in \mathcal{D}$ be a collection of nodes resulting from a length $l$ random walk in the {\bf NeighborhoodSearch} procedure. Define the $k$-length contexts for node $u_i$ as the nodes neighborhing it in a $k$-sized window $u_{i-k}$, $\ldots$, $u_{i-1}$, $u_{i+1}$, $\ldots$, $u_{i+k}$ and let $c$ denote the collection of contexts for $w$. Let $\#(w,c)$ denote the number of times the node-context pair $(w,c)$ appears in $\mathcal{D}$. Further, let $\#(w)$ and $\#(c)$ denote the number of times the node $w$ and the context $c$ appear in $\mathcal{D}$, respectively. As shown in \cite{levy2014neural}, running Skip-gram with negative sampling is equivalent to implicitly factorizing

\begin{equation}\label{eq:factorizing}
	\log\left(\dfrac{\#(w,c) |\mathcal{D}|}{\#(w) \#(c)} \right) - \log(b), \end{equation}

\noindent where $b$ is the number of negative samples specified. Expression (\ref{eq:factorizing}) suggests that by getting a hold of the quantity in the first logarithm of the expression, we can relate multi-node2vec directly to matrix factorization. 

Our results provide asymptotic expressions for ${\#(w,c) |\mathcal{D}|}/{\#(w) \#(c)}$ when the random walk length $l \rightarrow \infty$. To make our result explicit, we need to first introduce a little more notation.  Define $\widetilde{\mathbf{d}}_u = \sum_{v \in \mathcal{N}} \widetilde{A}_{u,v}(r)$ as the generalized degree of node $u$ in $\widetilde{\mathbf{G}}_{\mathcal{N}}{(r)}$ and let $\widetilde{\mathbf{D}} = \text{diag}(\widetilde{\mathbf{d}}_1, \ldots, \widetilde{\mathbf{d}}_N)$. Define the volume of $\mathbf{G}_\mathcal{N}(r)$ as ${\text{vol}}(\widetilde{\mathbf{G}}_\mathcal{N}(r)) = \sum_{u \in \mathcal{N}} \widetilde{\mathbf{d}}_u$. Define $\underline{\mathbf{P}}$ as the array containing the second order transition probabilities of {\bf NeighborhoodSearch}: $\underline{\mathbf{P}} = \{\underline{{P}}_{u,v,w} = P(u_{j+1} = u \mid u_j = v, u_{j-1} = w)\}$ and let ${\mathbf{X}}$ be its corresponding stationary distribution satisfying $\sum_w \underline{P}_{u,v,w} X_{v,w} = X_{u,v}$. Furthermore, let $\underline{P}^k_{u,v,w} = P(u_{j+r} = u \mid u_j = v, u_{j-1} = w)\}$ denote the $k$th step transition probability.


Finally, suppose $\stackrel{P}{\rightarrow}$ denotes convergence in probability. Our analysis of multi-node2vec depend on the bias of the transition probabilities for the random walks of the {\bf NeighborhoodSearch} procedure in equation (\ref{eq:transition_probs}), $\alpha_{pqr}(t,x,\ell, \ell')$. We can now state our next theorem, which relates multi-node2vec directly with matrix factorization. 

\begin{theorem}\label{thm:deepwalk}
	Let $\mathbf{G}_{\mathcal{N}}^m$ be an observed multilayer network and let $\widetilde{\mathbf{G}}_{\mathcal{N}}{(r)}$ be its adjusted aggregate network. Suppose that $\widetilde{\mathbf{G}}_{\mathcal{N}}{(r)}$ is connected, undirected, and non-bipartite. Let $k$ be the context size chosen for the {\bf Optimization} procedure. Then as $l \rightarrow \infty$,
	\begin{enumeratea}
		\item For all $p, q, r > 0$, 
		\begin{equation}\label{eq:res1}\dfrac{\#(w,c) |\mathcal{D}|}{\#(w) \#(c)} \stackrel{P}{\rightarrow} \dfrac{1}{2k}\dfrac{\sum_{j = 1}^k \left(\sum_u X_{w,u}\underline{P}^j_{c,w,u} + \sum_{u} X_{c,u} \underline{P}^j_{w,c,u}\right)}{\left(\sum_u X_{w,u} \right)\left(\sum_u X_{c,u} \right)}
			\end{equation}
		\item Let $\widetilde{\mathbf{P}} = \widetilde{\mathbf{D}}^{-1}\widetilde{\mathbb{A}}$. If $p = q = 1$,
		\begin{equation}\label{eq:res2}
			\dfrac{\#(w,c) |\mathcal{D}|}{\#(w) \#(c)} \stackrel{P}{\rightarrow} \dfrac{{\text{vol}}(\widetilde{\mathbf{G}}_\mathcal{N}(r))}{k}\left(\sum_{x = 1}^k \widetilde{\mathbf{P}}^k\right)\widetilde{\mathbf{D}}^{-1} \end{equation}
			\noindent for all $r > 0$.
	\end{enumeratea}

\end{theorem}

By applying the result of Lemma \ref{thm:equiv}, we can apply Theorems 2.1 - 2.3 and result (8) from \cite{qiu2018network} directly to prove the Theorem 3. Results (\ref{eq:res1}) and (\ref{eq:res2}) provide closed form limiting expressions for the matrix factorization problem in (\ref{eq:factorizing}). These results suggest the use of matrix factorization to identify features for a multilayer network; however, it should be noted that calculating and storing the second order transition probabilities $\underline{\mathbf{P}}$ and its stationary distribution $\mathbf{X}$ is computationally prohibitive. We do not consider such an algorithm in our current study but plan to address fast matrix factorization in future work.

\section{Discussion}

In this paper, we introduced the multi-node2vec algorithm, a fast network embedding technique for complex multilayer networks. This work motivates several areas of future work. For example, an important next step is to incorporate partial supervision for the detection of relevant features that depend on the application under investigation. Recent work like \cite{kipf2016semi} for semi-supervised feature engineering on static networks may provide a principled first step in the investigation for multilayer networks. We furthermore believe that it will be fruitful to thoroughly compare and contrast feature engineering methods like multi-node2vec with the results of multilayer community detection methods so as to better understand the discovered features. Furthermore, though not explicitly considered here, multi-node2vec is readily applicable to dynamic networks, an example of multilayer networks where the ordering of layers depends on time. This work will require incorporating appropriate notions of conditional dependence between the layers that replace the conditional independence assumptions applied here. Finally, our theoretical analysis of the multi-node2vec algorithm motivates further work in understanding the relationship between neural network algorithms with more traditional machine learning tasks such as matrix factorization. We believe that more work should be done in this area to fully understand the theoretical underpinnings of deep learning.

There are several aspects of the multi-node2vec algorithm that open up new directions of research. First, there is a need for an embedding procedure for dynamic networks that incorporates dependencies between each observed network in a sequence. This is particularly of interest for functional connectivity data observed through time, like the raw data considered in this paper. With multi-node2vec, a dynamic generalization is possible through the relaxing of simplifying assumptions (A1) and (A2) to allow for dependence between neighborhoods across time. Alternatively, one could directly construct an embedding algorithm for dynamic generative models like the family of temporal latent space models \cite{sewell2015latent} or temporal exponential random graph models \cite{hanneke2010discrete, lee2020varying}. Second, in this paper we treated a population as a collection of people across scans. However, in many instances multiple scans, like different tasks for example, are available for each individual. It would be interesting and perhaps very fruitful to obtain a multilayer embedding for each indivual instead using a multilayer network collected across all the available scans. Finally, much work has recently focused on mixture models for populations of networks, where the networks themselves cluster. One could account for such structure in multi-node2vec by again reworking the assumptions in (A1) and (A2) to account for differences between clusters of networks.

The multi-node2vec technique has potential for ground-breaking discovery in the study of functional connectivity. By specifying a multilayered framework that (i) models weighted networks, (ii) does not require temporal ordering of the layers, and (iii) is robust to noisy layers, multi-node2vec enables the study of networks that vary across individuals and cognitive tasks. Neuroscientists have very recently begun to utilize multilayer analyses \cite{betzel2016multi, bassett2015learning, bassett2011dynamic, muldoon2016network, braun2015dynamic}. The majority of this work has explored how community structure and network modules vary across time. For instance, one study showed that shifts in community structure across time predict differences in learning a visual-motor task \cite{bassett2015learning}. Indeed, network neuroscientists have lately called for greater emphasis on multilayer techniques, particularly those that do not require temporal ordering of layers, thus allowing for more comprehensive quantification of networks across samples \cite{muldoon2016network}. The multi-node2vec algorithm is a fully data-driven strategy with the capabilities to learn significant neurological variation among brains, and will progress the investigation of individual differences and disease.

\section*{Funding}

JDW gratefully acknowledges support on this project by the National Science Foundation grant NSF DMS-1830547.

\section*{Conflict of Interest}

All authors declare that are no conflicts of interest with the work presented in this manuscript.

%
%
%

\bibliographystyle{nws}


\section*{Appendix: Additional Results for Subnetwork Classification Study}

Below, we provide the AUC across competing methods for classification of functional subnetworks in healthy individuals. These results illustrate the results when the methods are applied across the population of healthy individuals (Figure \ref{fig:classification}), and when applied to the population of healthy individuals with 10 layers of noise added to the network (Figure \ref{fig:classification10}). These results complement those already provided and discussed in Section 4.3.
\begin{figure}[ht]
\centering
\includegraphics[width = .8\textwidth, angle = 270]{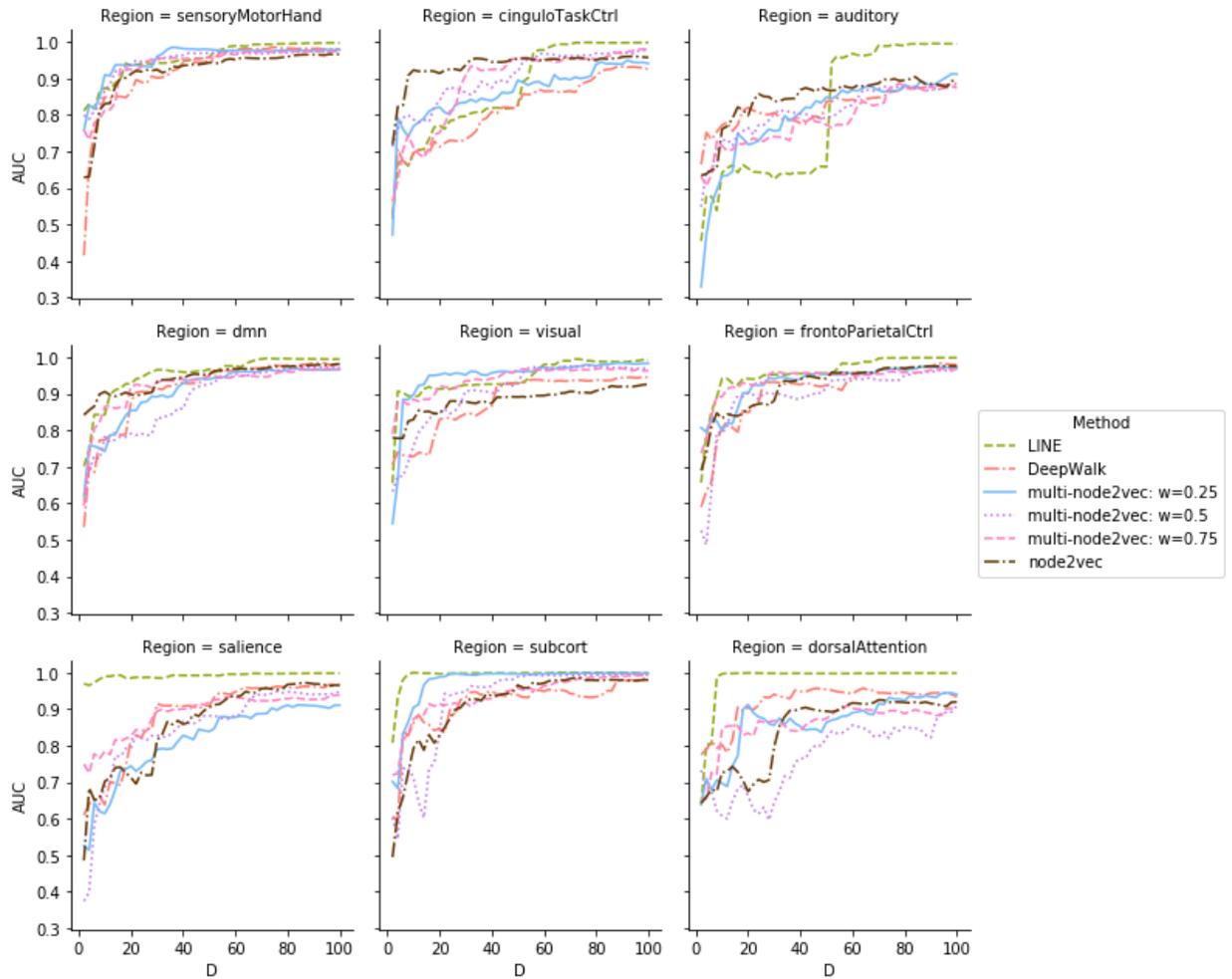}
\caption{The AUC of a one vs. all logistic regression classifier for the nine major functional subnetworks of the brain across 74 healthy individuals. Plots show the AUC of the classifier against the number of dimensions $D$ for feature representations from multi-node2vec, node2vec, DeepWalk, LINE, and the spectral decomposition. \label{fig:classification}}
\end{figure}

\begin{figure}[ht]
\centering
\includegraphics[width = .8\textwidth, angle = 270]{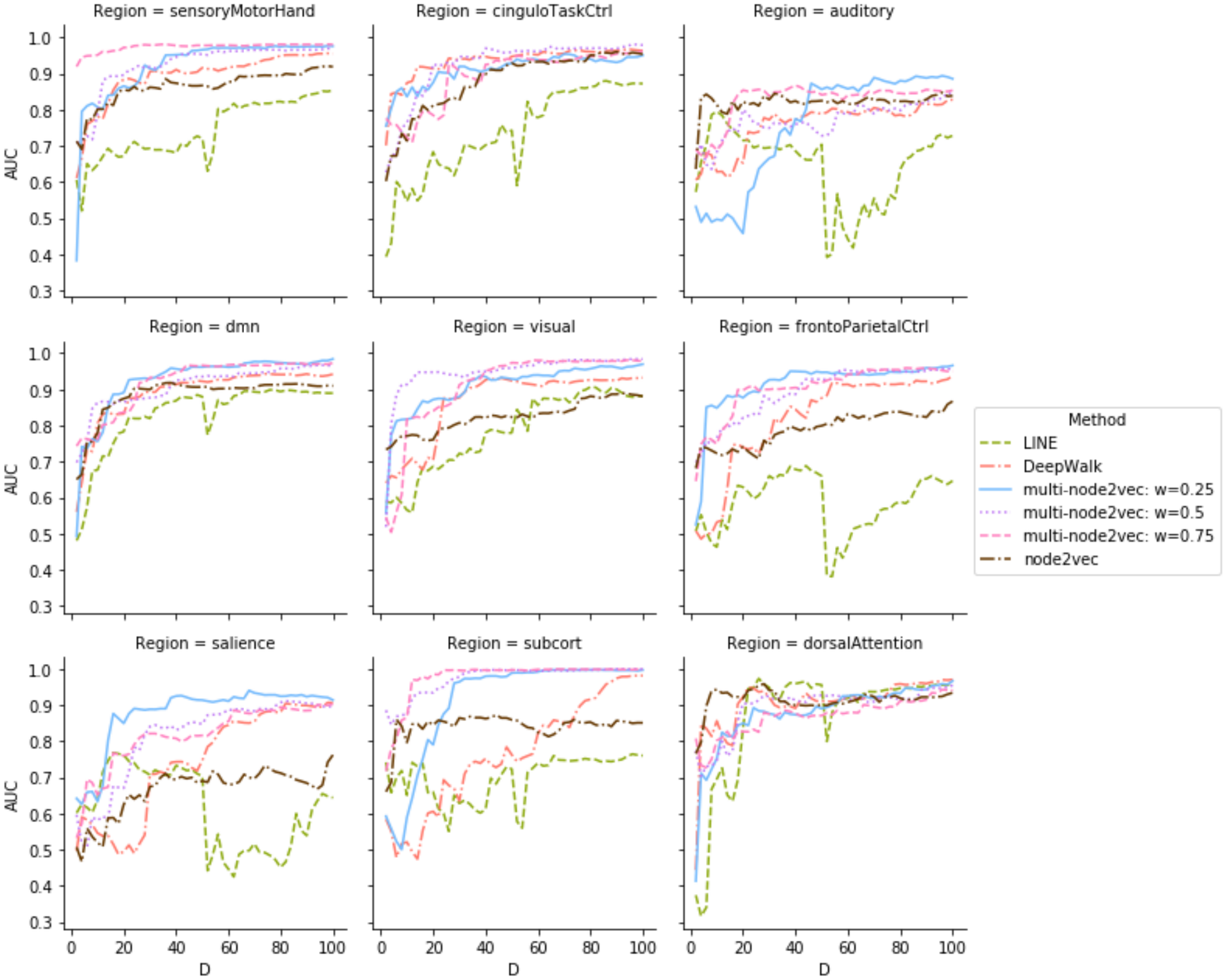}
\caption{The AUC of a one vs. all logistic regression classifier for the nine major functional subnetworks of the brain across all 74 healthy individuals and 10 layers of noise. Plots show the AUC of the classifier against the number of dimensions $D$ for feature representations from multi-node2vec, node2vec, DeepWalk, LINE and spectral decomposition. \label{fig:classification10}}
\end{figure}

\end{document}